\documentclass[a4paper,10pt]{scrartcl}
\usepackage[T1]{fontenc}
\usepackage[latin1]{inputenc}
\usepackage{float}
\usepackage{amsmath}
\usepackage{url}
\usepackage{tikz}
\usetikzlibrary{decorations.pathreplacing}
\usepackage{algorithm}
\usepackage[noend]{algpseudocode}
\usepackage{authblk}
\usepackage{amsthm}

\newtheorem{theorem}{Theorem}
\newtheorem{lemma}[theorem]{Lemma}

\newcommand*\Let[2]{\State #1 $\gets$ #2}

\newcommand{\splitmem}{splitMEM}
\newcommand{\algoA}{\textsf{A1}}
\newcommand{\algoB}{\textsf{A2}}
\newcommand{\algoC}{\textsf{A3}}
\newcommand{\algoCone}{\textsf{A3compr1}}
\newcommand{\algoCtwo}{\textsf{A3compr2}}
\newcommand{\algoD}{\textsf{A4}}
\newcommand{\algoDone}{\textsf{A4compr1}}
\newcommand{\algoDtwo}{\textsf{A4compr2}}
\newcommand{\algoDplus}{\textsf{A4+explicit}}
\newcommand{\algoDoneplus}{\textsf{A4compr1+explicit}}
\newcommand{\algoDtwoplus}{\textsf{A4compr2+explicit}}
\newcommand{\hd}{\hphantom{4}}
\newcommand{\seqnumber}{d}
\newcommand{\suffixlb}{\mathit{suffix}\_lb}
\newcommand{\BVl}{B_l}
\newcommand{\BVr}{B_r}
\newcommand{\LeftBoundary}{lb}
\newcommand{\RightBoundary}{rb}
\newcommand{\SNT}{\mathtt{\symbol{36}}}
\newcommand{\lcp}[0]{\mathsf{lcp}}
\newcommand{\LCP}[0]{\mathsf{LCP}}
\newcommand{\SUF}[0]{\mathsf{SA}}
\newcommand{\SA}[0]{\mathsf{SA}}
\newcommand{\LF}[0]{\mathsf{LF}}
\newcommand{\BWT}{\mathsf{BWT}}
\newcommand{\SF}[1]{S_{\SUF[#1]}}
\newcommand{\Keyw}[1]{{\textbf{#1}}}

\begin{document}

\title{A representation of a compressed de Bruijn graph for pan-genome analysis that enables search}
\author{Timo Beller}
\author{Enno Ohlebusch}
\affil{Institute of Theoretical Computer Science, Ulm University, D-89069 Ulm, Germany \texttt{\{Timo.Beller,Enno.Ohlebusch\}@uni-ulm.de}}
\date{}
\maketitle

\begin{abstract}
Recently, Marcus et al.\ (Bioinformatics 2014) proposed to use a compressed de
Bruijn graph to describe the relationship between
the genomes of many individuals/strains of the same or closely related species.
They devised an $O(n\log g)$ time algorithm called splitMEM that
constructs this graph directly (i.e., without using the
uncompressed de Bruijn graph) based on a suffix tree, where $n$ is the total
length of the genomes and $g$ is the length of the longest genome.
In this paper, we present a construction algorithm
that outperforms their algorithm in theory and in practice.
Moreover, we propose a new space-efficient representation of the
compressed de Bruijn graph that adds the possibility to search
for a pattern (e.g.\ an allele---a variant form of a gene) within
the pan-genome.
\end{abstract}

\section{Introduction}
\label{sec:Introduction}
Today, next generation sequencers produce vast amounts of DNA sequence
information and it is often the case that multiple genomes of the same or
closely related species are available. An example is the 1000 Genomes Project,
which started in 2008. Its goal was to
sequence the genomes of at least 1000 humans from all over the world and to
produce a catalog of all variations (SNPs, indels, etc.) in the human
population. The genomic sequences together with this catalog is called the
``pan-genome'' of the population. There are several approaches that try to
capture variations between many individuals/strains in a population graph; see
e.g.\ \cite{SCH:HAG:OSS:2009,HUA:POP:BAT:2013,RAH:WEE:REI:2014,DIL:ETAL:2015}.
These works all require a multi-alignment as input. By contrast,
Marcus et al.\ \cite{MAR:LEE:SCH:2014} use a compressed de Bruijn graph
of maximal exact matches (MEMs) as a graphical representation of the
relationship between genomes;
see Section \ref{sec-Compressed de Bruijn graph} for
a definition of de Bruijn graphs. They describe an $O(n\log g)$ time
algorithm that directly computes the compressed de Bruijn graph
on a suffix tree, where $n$ is the total length of the genomes
and $g$ is the length of the longest genome.
Marcus et al.\ write in \cite[Section 4]{MAR:LEE:SCH:2014}:
``Future work remains to improve splitMEM and further unify the family of sequence indices. Although ..., most desired are techniques to reduce the space consumption ...''
In this article, we present such a technique. To be more precise,
we will develop an $O(n\log \sigma)$ time algorithm that
constructs the compressed de Bruijn graph directly on an FM-index
of the genomes, where $\sigma$ is the size of the underlying alphabet.
This algorithm is faster than the algorithms described in a preliminary
version of this article \cite{BEL:OHL:2015}.
Moreover, we propose a new space-efficient representation of the
compressed de Bruijn graph that adds the possibility to search
for a pattern (e.g.\ an allele---a variant form of a gene) within
the pan-genome. More precisely, one can use the FM-index to search
for the pattern and, if the pattern occurs in the
pan-genome, one can start the exploration of the compressed de Bruijn graph
at the nodes that correspond to the pattern.

The \emph{contracted}
de Bruijn graph introduced by Cazaux et al.\ \cite{CAZ:LEC:RIV:2013}
is closely related but not identical to the compressed de Bruijn graph.
A node in the contracted de Bruijn graph is not necessarily a substring
of one of the genomic sequences (see the remark following Definition 3 in
\cite{CAZ:LEC:RIV:2013}). Thus the contracted de Bruijn graph, which can
be constructed in linear time from the suffix tree \cite{CAZ:LEC:RIV:2013},
is not useful for our purposes.

\section{Preliminaries}
\label{sec-Preliminaries}

Let $\Sigma$ be an ordered alphabet of size $\sigma$ whose smallest element is
the sentinel character $\SNT$. In the following,
$S$ is a string of length $n$ on $\Sigma$ having the sentinel character
at the end (and nowhere else).
In pan-genome analysis, $S$ is the concatenation of multiple
genomic sequences, where the different sequences are separated by special
symbols (in practice, we use one separator symbol and treat the
different occurrences of it as if they were different characters;
see Section \ref{sec-Computation of right-maximal k-mers}).
For $1 \leq i \leq n$, $S[i]$ denotes the \emph{character at position} $i$ in
$S$. For $i \leq j$, $S[i..j]$ denotes the \emph{substring} of $S$ starting
with the character at position $i$ and ending with the character at position
$j$. Furthermore, $S_i$ denotes the $i$-th suffix $S[i..n]$ of $S$.
The \emph{suffix array} $\SUF$ of the string $S$
is an array of integers in the range $1$ to $n$ specifying the
lexicographic ordering of the $n$ suffixes of $S$,
that is, it satisfies $\SF{1} < \SF{2} < \cdots < \SF{n}$;
see Fig.\ \ref{fig:suffix array} for an example.
A suffix array can be constructed in linear time; see e.g.\
the overview article \cite{PUG:SMY:TUR:2007}.
For every substring $\omega$ of $S$, the $\omega$-interval
is the suffix array interval $[i..j]$ so that $\omega$ is a
prefix of $S_{\SUF[k]}$ if and only if $i\leq k \leq j$.

The Burrows-Wheeler transform \cite{BUR:WHE:1994}
converts $S$ into the string
$\BWT[1..n]$ defined by $\BWT[i]=S[\SUF[i] -1]$ for all
$i$ with $\SUF[i] \neq 1$ and $\BWT[i] = \SNT$ otherwise;
see Fig.\ \ref{fig:suffix array}.
Several semi-external and external memory algorithms are known
that construct the $\BWT$ directly (i.e., without constructing
the suffix array); see e.g.\ 
\cite{KAR:2007,OKA:SAD:2009,FER:GAG:MAN:2010,BEL:ZWE:GOG:OHL:2013}.

\begin{figure}[ht]
\begin{center}
\begin{scriptsize}
\begin{tabular}{|r|r|r|c|c|r|r|c|l|}
\hline
$i$ & $\SA$       & $\LCP$      & $\BVr$ & $\BVl$ & $\LF$ & $\Psi$ & $\BWT$      & $S_{\SA[i]}$              \\ \hline
\texttt{ 1} & \texttt{15} & \texttt{-1} & \texttt{ 0 } & \texttt{ 0 } & \texttt{10}& \texttt{5 } & \texttt{G}  & \texttt{\$}               \\ \hline
\texttt{ 2} & \texttt{12} & \texttt{ 0} & \texttt{ 1 } & \texttt{ 0 } & \texttt{13}& \texttt{6 } & \texttt{T}  & \texttt{ACG\$}            \\ \hline
\texttt{ 3} & \texttt{ 8} & \texttt{ 3} & \texttt{ 0 } & \texttt{ 0 } & \texttt{14}& \texttt{7 } & \texttt{T}  & \texttt{ACGTACG\$}        \\ \hline
\texttt{ 4} & \texttt{ 4} & \texttt{ 7} & \texttt{ 1 } & \texttt{ 0 } & \texttt{15}& \texttt{8 } & \texttt{T}  & \texttt{ACGTACGTACG\$}    \\ \hline
\texttt{ 5} & \texttt{ 1} & \texttt{ 2} & \texttt{ 0 } & \texttt{ 0 } & \texttt{1}& \texttt{9 } & \texttt{\$} & \texttt{ACTACGTACGTACG\$} \\ \hline
\texttt{ 6} & \texttt{13} & \texttt{ 0} & \texttt{ 0 } & \texttt{ 0 } & \texttt{2}& \texttt{10 } & \texttt{A}  & \texttt{CG\$}             \\ \hline
\texttt{ 7} & \texttt{ 9} & \texttt{ 2} & \texttt{ 0 } & \texttt{ 0 } & \texttt{3}& \texttt{11 } & \texttt{A}  & \texttt{CGTACG\$}         \\ \hline
\texttt{ 8} & \texttt{ 5} & \texttt{ 6} & \texttt{ 0 } & \texttt{ 0 } & \texttt{4}& \texttt{12 } & \texttt{A}  & \texttt{CGTACGTACG\$}     \\ \hline
\texttt{ 9} & \texttt{ 2} & \texttt{ 1} & \texttt{ 0 } & \texttt{ 1 } & \texttt{5}& \texttt{15 } & \texttt{A}  & \texttt{CTACGTACGTACG\$}  \\ \hline
\texttt{10} & \texttt{14} & \texttt{ 0} & \texttt{ 0 } & \texttt{ 0 } & \texttt{6}& \texttt{1 } & \texttt{C}  & \texttt{G\$}              \\ \hline
\texttt{11} & \texttt{10} & \texttt{ 1} & \texttt{ 0 } & \texttt{ 0 } & \texttt{7}& \texttt{13 } & \texttt{C}  & \texttt{GTACG\$}          \\ \hline
\texttt{12} & \texttt{ 6} & \texttt{ 5} & \texttt{ 0 } & \texttt{ 1 } & \texttt{8}& \texttt{14 } & \texttt{C}  & \texttt{GTACGTACG\$}      \\ \hline
\texttt{13} & \texttt{11} & \texttt{ 0} & \texttt{ 0 } & \texttt{ 0 } & \texttt{11}& \texttt{2 } & \texttt{G}  & \texttt{TACG\$}           \\ \hline
\texttt{14} & \texttt{ 7} & \texttt{ 4} & \texttt{ 0 } & \texttt{ 0 } & \texttt{12}& \texttt{3 } & \texttt{G}  & \texttt{TACGTACG\$}       \\ \hline
\texttt{15} & \texttt{ 3} & \texttt{ 8} & \texttt{ 0 } & \texttt{ 0 } & \texttt{9}& \texttt{4 } & \texttt{C}  & \texttt{TACGTACGTACG\$}   \\ \hline
\texttt{16} & \texttt{  } & \texttt{-1} & \texttt{   } & \texttt{   } & \texttt{ } & \texttt{ } & \texttt{ } & \texttt{ } \\ \hline
\end{tabular}
\end{scriptsize}
\end{center}
\caption{The suffix array $\SUF $ of the string \texttt{ACTACGTACGTACG\$} and
related notions are defined in Section \ref{sec-Preliminaries}.
The bit vectors $\BVr$ and $\BVl$ are explained in Section \ref{sec-Computation of right-maximal k-mers}.}
\label{fig:suffix array}
\end{figure}

The \emph{wavelet tree} \cite{GRO:GUP:VIT:2003} of the $\BWT$ supports
one backward search step in $O(\log \sigma)$ time \cite{FER:MAN:2000}:
Given the $\omega$-interval $[lb..rb]$ and a character $c\in \Sigma$,
$backwardSearch(c,[lb..rb])$ returns the $c\omega$-interval $[i..j]$
(i.e., $i\leq j$ if $c\omega$ is a substring of $S$; otherwise $i > j$).
This crucially depends on the fact
that a bit vector $B$ can be preprocessed in linear time
so that an arbitrary $rank_1(B,i)$ query (asks for the number of ones in
$B$ up to and including position $i$) can be answered in constant time
\cite{JAC:1989}. Backward search can be generalized on the wavelet tree as
follows: Given an $\omega$-interval $[lb..rb]$, a slight modification of
the procedure $getIntervals([lb..rb])$ described in \cite{BEL:GOG:OHL:SCH:2013}
returns the list $[(c,[i..j]) \mid c\omega \mbox{ is a substring of } S
\mbox{ and } [i..j] \mbox{ is the } c\omega\mbox{-interval}]$,
where the first component of an element $(c,[i..j])$ must be a character.
The worst-case time complexity of the procedure $getIntervals$ is
$O(z + z \log (\sigma/z))$, where $z$ is the number of elements in
the output list; see \cite[Lemma 3]{GAG:NAV:PUG:2012}.

The $\LF$-mapping (last-to-first-mapping)
is defined as follows: If $\SUF[i] = q$, then
$\LF(i)$ is the index $j$ so that $\SUF[j] = q-1$ (if $\SUF[i] = 1$,
then $\LF(i)=1$). In other words, if the $i$-th entry in the suffix array
is the suffix $S_q$, then $\LF(i)$ ``points'' to the entry at which
the suffix $S_{q-1}$ can be found; see Fig.~\ref{fig:suffix array}.
The function $\Psi$ is the inverse of the $\LF$-mapping.
Using the wavelet tree of the $\BWT$, a value $\LF(i)$ or $\Psi(i)$
can be calculated in $O(\log \sigma)$ time.
For later purposes, we recall how the $\LF$-mapping can be computed
from the $\BWT$. First, the $C$-array is calculated, where
for each $c \in \Sigma$, $C[c]$ is the overall number of occurrences of
characters in $\BWT$ that are strictly smaller than $c$.
Second, if in a left-to-right scan of the $\BWT$,
where the loop-variable $i$ varies from $1$ to $n$,
$C[c]$ is incremented by one for $c = \BWT[i]$, then $\LF[i] = C[c]$.

The suffix array $\SUF$ is often enhanced with the so-called
$\LCP$-array containing the lengths of longest common prefixes between
consecutive suffixes in $\SUF$; see Fig.\ \ref{fig:suffix array}.
Formally, the  $\LCP$-array is an array so that $\LCP[1] = -1 = \LCP[n+1]$ and
$\LCP[i] = |\lcp(\SF{i-1},\SF{i})|$ for $2\leq i \leq n$,
where $\lcp(u,v)$ denotes the longest common prefix
between two strings $u$ and $v$. The $\LCP$-array can be computed in linear
time from the suffix array and its inverse, but it is also possible to construct it directly from the wavelet tree of
the $\BWT$ in $O(n \log \sigma)$ time with the
help of the procedure $getIntervals$ \cite{BEL:GOG:OHL:SCH:2013}.

A substring $\omega$ of $S$ is a \emph{repeat} if it occurs at least twice
in $S$. Let $\omega$ be a repeat of length $\ell$ and let $[i..j]$ be the
$\omega$-interval. The repeat $\omega$ is \emph{left-maximal} if
$|\{\BWT[x] \mid i \leq x \leq j\}| \geq 2$, i.e., the set
$\{S[\SUF[x]-1] \mid i \leq x \leq j\}$ of all characters that precede
at least one of the suffixes $\SF{i},\dots,\SF{j}$ is not singleton
(where $S[0] := \$ $).
Analogously, the repeat $\omega$ is \emph{right-maximal}
if $|\{S[\SUF[x]+\ell] \mid i \leq x \leq j\}| \geq 2$.
A left- and right-maximal repeat is called \emph{maximal} repeat.
(Note that \cite{MAR:LEE:SCH:2014} use the term ``maximal exact match''
instead of the more common term ``maximal repeat''. We will not use
the term ``maximal exact match'' here.)
A detailed explanation of the techniques used here can be found in
\cite{OHL:2013}.

\section{Compressed de Bruijn graph}
\label{sec-Compressed de Bruijn graph}

Given a string $S$ of length $n$ and a natural number $k$,
the de Bruijn graph of $S$
contains a node for each distinct length $k$ substring of $S$,
called a $k$-mer. Two nodes $u$ and $v$ are connected by a directed edge
$(u,v)$ if $u$ and $v$ occur consecutively in $S$,
i.e., $u = S[i..i+k-1]$ and $v = S[i+1..i+k]$.
Fig.\ \ref{fig:de Bruijn graph} shows an example. Clearly,
the graph contains at most $n$ nodes and $n$ edges.
By construction, adjacent nodes will overlap by $k-1$ characters,
and the graph can include multiple edges connecting the same pair
of nodes or self-loops representing overlapping repeats.
For every node, except for the start node (containing the first $k$ characters
of $S$) and the stop node (containing the last $k$ characters of $S$),
the in-degree coincides with the out-degree.
A de Bruijn graph can be ``compressed'' by merging non-branching chains of
nodes into a single node with a longer string. More precisely,
if node $u$ is the only predecessor of node $v$ and $v$ is the only
successor of $u$ (but there may be multiple edges $(u,v)$), then
$u$ and $v$ can be merged into a single node that has the predecessors of
$u$ and the successors of $v$. After maximally compressing the graph,
every node (apart from possibly the start node) has at least two different
predecessors or its single predecessor has at least two different successors
and every node (apart from the stop node) has at least two different
successors or its single successor has at least two different predecessors;
see Fig.\ \ref{fig:de Bruijn graph}.
Of course, the compressed de Bruijn graph can be
built from its uncompressed counterpart (a much larger graph), but this is
disadvantageous because of the huge space consumption. That is why we will
build it directly.

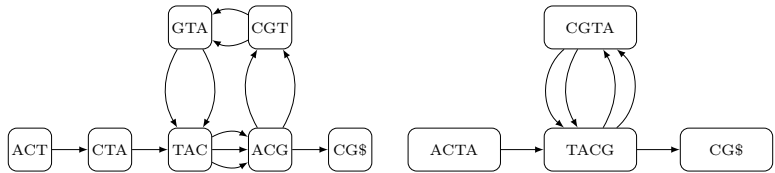
\begin{figure}[H]
\begin{center}
\scalebox{0.81}{
\begin{tikzpicture}[
x=1.30cm, y=2.00cm,
every path/.style={-latex},
every node/.style={draw,inner sep=0pt,minimum size=20pt,font=\scriptsize},
uncompressed/.style={rounded corners, },
compressed/.style={rounded corners, text width=1.5cm, align=center},
]
\node[uncompressed] (1) at (0,0) {ACT} ;
\node[uncompressed] (2) at (1,0) {CTA} ;
\node[uncompressed] (3) at (2,0) {TAC} ;
\node[uncompressed] (4) at (3,0) {ACG} ;
\node[uncompressed] (5) at (3,1) {CGT} ;
\node[uncompressed] (6) at (2,1) {GTA} ;
\node[uncompressed] (7) at (4,0) {CG\$};
\draw (1) to               (2);
\draw (2) to               (3);
\draw (3) to               (4);
\draw (3) to  [bend left]  (4);
\draw (3) to  [bend right] (4);
\draw (4) to  [bend left]  (5);
\draw (4) to  [bend right] (5);
\draw (4) to               (7);
\draw (5) to  [bend left]  (6);
\draw (5) to  [bend right] (6);
\draw (6) to  [bend left]  (3);
\draw (6) to  [bend right] (3);

\node[compressed] (A) at (5.3,0) {ACTA};
\node[compressed] (B) at (7.0,0) {TACG};
\node[compressed] (C) at (7.0,1) {CGTA};
\node[compressed] (D) at (8.7,0) {CG\$};
\draw (A) to                   (B);
\draw (B) to                   (D);
\draw (B) to [out=40, in=-40]  (C);
\draw (B) to [bend right] (C);
\draw (C) to [bend right]  (B);
\draw (C) to [out=-140, in=140] (B);
\end{tikzpicture}
}
\end{center}
\caption{The de Bruijn graph for $k=3$ and the string ACTACGTACGTACG\$ is shown
on the left, while its compressed counterpart is shown on the right.}
\label{fig:de Bruijn graph}
\end{figure}

Fig.\ \ref{fig:graph G} shows how splitMEM represents the compressed de Bruijn
graph $G$ for $k=3$ and the string $S=$ ACTACGTACGTACG\$.
Each node corresponds to a substring $\omega$ of $S$ and
consists of the components $(id,len,posList,adjList)$,
where $id$ is a natural number that uniquely identifies the node,
$len$ is the length $|\omega|$ of $\omega$,
$posList$ is the list of positions at which $\omega$ occurs in $S$
(sorted in ascending order), and $adjList$ is the list of the successors
of the node (sorted in such a way that the walk through
$G$ that gives $S$ is induced by the adjacency lists: if node $G[id]$ is
visited for the $i$-th time, then its successor is the node that can be
found at position $i$ in the adjacency list of $G[id]$).

\begin{figure}[H]
\[
\begin{array}{|c|c|c|c|c|}
\hline
id &len& posList & adjList& \omega\\\hline
1 & 4 & [5,9]   & [2,2]&\texttt{CGTA} \\\hline
2 & 4 & [3,7,11]& [1,1,4]&\texttt{TACG} \\\hline
3 & 4 & [1]     & [2]&\texttt{ACTA}  \\\hline
4 & 3 & [13]    & [~]&\texttt{CG\$} \\\hline
\end{array}
\]
\caption{Explicit representation of the compressed de Bruijn graph
from Fig.\ \ref{fig:de Bruijn graph}.
\label{fig:graph G}}
\end{figure}

The nodes in the compressed de Bruijn graph of a pan-genome
can be categorized as follows:
\begin{itemize}
\item a uniqueNode represents a unique substring
in the pan-genome and has a single start position
(i.e., $posList$ contains just one element)
\item a repeatNode represents a substring that occurs at least twice in the
pan-genome, either as a repeat in a single genome or as a segment shared by
multiple genomes.
\end{itemize}
In pan-genome analysis, $S$ is the concatenation of multiple
genomic sequences, where the different sequences are separated
by a special symbol $\#$.
(In theory, one could use pairwise different symbols to separate the
sequences, but in practice this would blow up the alphabet.)
This has the effect that $\#$ may be part of a repeat.
In contrast to splitMEM, our algorithm treats the different occurrences of
$\#$ as if they were different characters. Consequently,
$\#$ will not be a part of a repeat. In our approach, each occurrence of
$\#$ will be the end of a stop node (i.e., there is a
stop node for each sequence).

According to \cite{MAR:LEE:SCH:2014}, the compressed de Bruijn graph
is most suitable for pan-genome analysis:
``This way the complete pan-genome will be represented in a compact
graphical representation such that the shared/strain-specific status
of any substring is immediately identifiable, along with the context
of the flanking sequences. This strategy also enables powerful
topological analysis of the pan-genome not possible from a linear
representation.'' It has one defect though: it is not possible to
search efficiently for certain nodes and then to explore the graph in the
vicinity of these nodes. A user might, for example, want to search for
a certain allele in the pan-genome and---if it is present---to examine
the neighborhood of that allele in the graph.
Here, we propose a new space-efficient representation of the
compressed de Bruijn graph that adds exactly this functionality.

We store the graph in an array $G$ of length $N$, where
$N$ is the number of nodes in the compressed de Bruijn graph.
Moreover, we assign to each node a unique identifier $id \in \{1,\dots,N\}$.
A node $G[id]$ now has the form $(len,lb,size,\suffixlb)$, where
\begin{itemize}
\item $len$ is the length of the string $\omega = S[\SA[lb]..\SA[lb]+len-1]$
that corresponds to the node with identifier $id$
\item $[lb..lb+size-1]$ is the $\omega$-interval and $size$ is the size of the $\omega$-interval
\item $[\suffixlb..\suffixlb +size-1]$ is the interval of the $k$ length suffix of $\omega$
\end{itemize}
There is one exception though: the sentinel $\$$ and each occurrence of
the separator $\#$ will be the end of a stop node. Clearly, the suffix $\$$
of $S$ appears at index $1$ in the suffix array because $\$$ is the smallest
character in the alphabet. The suffix array interval of $\$$ is $[1..1]$,
so we set $\suffixlb = 1$. Analogously, a suffix of $S$ that starts with $\#$
appears at an index $j \in \{2,\dots,\seqnumber\}$ in the suffix array (where
$\seqnumber$ is the number of sequences in $S$) because $\#$ is the second
smallest character in the alphabet, so we set $\suffixlb = j$.

Fig.\ \ref{fig:space-efficient representation} shows an example.
Henceforth this representation will be called implicit representation,
while the representation from Fig.\ \ref{fig:graph G} will be called
explicit representation.
It is clear that in the implicit representation the list of all positions at
which $\omega$ occurs in $S$ can be computed as follows:
$[\SA[i] \mid lb \leq i \leq lb+size-1]$. It will be explained later,
how the graph can be traversed and how a pattern can be searched for.
We shall see that this can be done efficiently by means of the fourth
component $\suffixlb$.

\begin{figure}[H]
\[
\begin{array}{|c|c|c|c|c|c|}
\hline
id &len& lb &size & \suffixlb & \omega\\\hline
1 & 4 & 13 & 3 & 2 &\texttt{TACG}\\\hline
2 & 4 &  5 & 1 & 9 &\texttt{ACTA} \\\hline
3 & 4 &  7 & 2 & 11 &\texttt{CGTA}\\\hline
4 & 3 &  6 & 1 & 1 &\texttt{CG\$}\\\hline
\end{array}
\]
\caption{Implicit representation of the compressed de Bruijn graph
from Fig.\ \ref{fig:de Bruijn graph}.
\label{fig:space-efficient representation}}
\end{figure}

\section{Construction algorithm}
\label{sec-Construction of a compressed de Bruijn graph}

We will build the implicit representation of the compressed de Bruijn graph
directly from an FM-index (the wavelet tree of the $\BWT$) of $S$,
using Lemma \ref{lem-maximal repeat} (the simple proof is omitted).

\begin{lemma}
\label{lem-maximal repeat}
Let $v$ be a node in the compressed de Bruijn graph
and let $\omega$ be the string corresponding to $v$.
If $v$ is not the start node, then it has at least two different
predecessors if and only if the length $k$ prefix of $\omega$ is a
left-maximal repeat. It has at least two different successors if and only
if the length $k$ suffix of $\omega$ is a right-maximal repeat.
\end{lemma}

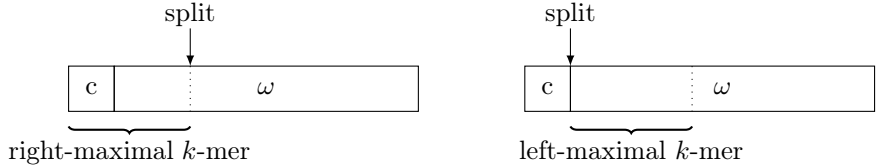
\begin{figure}[H]
\begin{tikzpicture}[
x=1.00cm, y=1.00cm,
]

\node[draw, text width=6mm, minimum height=6mm, inner sep=0mm, outer sep=0mm, align=center] (c1) at (0,0) {c};
\node[draw, text width=40mm, minimum height=6mm, inner sep=0mm, outer sep=0mm, align=center] (w1) at (23mm,0) {\(\omega\)};
\draw [thick, decoration={brace, mirror, raise=0.5cm}, decorate] (c1.west) to node [anchor=north, yshift=-0.55cm] {right-maximal $k$-mer} ([xshift=-10mm]w1);
\draw[dotted] ([xshift=-10mm]w1.north) to ([xshift=-10mm]w1.south);
\draw[-latex] ([yshift=5mm, xshift=-10mm]w1.north) to ([xshift=-10mm]w1.north) node[above, yshift=4mm] {split};

\node[draw, text width=6mm, minimum height=6mm, inner sep=0mm, outer sep=0mm, align=center] (c2) at (60mm,0) {c};
\node[draw, text width=40mm, minimum height=6mm, inner sep=0mm, outer sep=0mm, align=center] (w2) at (83mm,0) {\(\omega\)};
\draw [thick, decoration={brace, mirror, raise=0.5cm}, decorate] (c2.east) to node [anchor=north, yshift=-0.55cm] {left-maximal $k$-mer} ([xshift=-4mm]w2);
\draw[dotted] ([xshift=-4mm]w2.north) to ([xshift=-4mm]w2.south);
\draw[-latex] ([yshift=5mm, xshift=+3mm]c2.north) to ([xshift=+3mm]c2.north) node[above, yshift=4mm] {split};
\end{tikzpicture}
\caption{The string $\omega$ must be split if
the length $k$ prefix of $c\omega$ is a right-maximal repeat or
the length $k$ prefix of $\omega$ is a left-maximal repeat.}
\label{fig:splits}
\end{figure}

The general idea behind our algorithm is as follows. Compute the suffix array
intervals of all right-maximal $k$-mers. For each such $k$-mer
$v$, compute all $cv$-intervals, where $c\in \Sigma$.
Then, for each $u=cv$, compute all $bu$-intervals, where
$b\in \Sigma$, etc. In other words, we start with all right-maximal $k$-mers
and extend them as long as possible (and in all possible ways with
the procedure $getIntervals$),
character by character, to the left. According to Lemma
\ref{lem-maximal repeat}, the left-extension of a string
$\omega$ must stop if (i) the length $k$ prefix
of $\omega$ is a left-maximal repeat (this is the case if the
procedure $getIntervals$ applied to the $\omega$-interval returns a
non-singleton list). It must also stop if
(ii) the length $k$ prefix $v$ of $c\omega$ is a right-maximal repeat
for some $c\in \Sigma$; see Fig.\ \ref{fig:splits}.
This is because by Lemma \ref{lem-maximal repeat} there is a node $uv$,
$u\in \Sigma^*$, in the compressed de Bruijn graph with at least two different
successors (the length $k$ suffix $v$ of $uv$ is a right-maximal repeat).
Consequently, there must be a directed edge $(uv,\omega)$ in the
compressed de Bruijn graph. In the following, we will explain the different
phases of the algorithm in detail.

\subsection{Computation of right-maximal $k$-mers and node identifiers}
\label{sec-Computation of right-maximal k-mers}

The construction algorithm uses two bit vectors $\BVr$ and $\BVl$.
To obtain the bit vector $\BVr$, we must compute all right-maximal $k$-mers
and their suffix array intervals. Let $u$ be a right-maximal
$k$-mer and consider the $u$-interval $[lb..rb]$ in the suffix array.
Note that (1) $\LCP[lb] < k$ and (2) $\LCP[rb+1] < k$.
Since $u$ is right-maximal, $u$ is the longest common prefix of all suffixes
in the interval $[lb..rb]$. This implies
(3) $\LCP[j] \geq k$ for all $j$ with $lb+1\leq j \leq rb$ and
(4) $\LCP[j] = k$ for at least one $j$ with $lb+1\leq j \leq rb$
(in the terminology of \cite{ABO:KUR:OHL:2004},
$[lb..rb]$ is an lcp-interval of lcp-value $k$).
It follows as a consequence that the bit vector $\BVr$ can be calculated
with the help of the $\LCP$-array. Using the algorithm of
\cite{BEL:GOG:OHL:SCH:2013}, Algorithm \ref{alg:debruijnhelper} constructs
the $\LCP$-array directly from the $\BWT$ in $O(n \log \sigma)$ time, where $\sigma$ is the size of
the alphabet. It is not difficult to verify that lines
\ref{begin k-mers} to \ref{end k-mers} of Algorithm \ref{alg:debruijnhelper}
compute all suffix array intervals of right-maximal $k$-mers.
Furthermore, on lines \ref{set bitvector at lb} and \ref{end k-mers}
the boundaries $lb$ and $rb=i-1$ of the $k$-mer intervals are marked by
setting the entries of $\BVr$ at these positions to $1$. On line
\ref{add node to G}, the node $(lb,k,i-lb,lb)$ having the current value
of the variable $counter$ as identifier is added to the graph $G$.
In contrast to the last two components, the first two components of a node
may change later (they will change when a left-extension is possible).
On line \ref{add node to Q}, the node identifier is added to the queue
$Q$ and then $counter$ is incremented by one.

\begin{algorithm}[H]
  \caption{Construction of the bit vectors $\BVr$ and $\BVl$.
    \label{alg:debruijnhelper}}
  \begin{algorithmic}[1]
    \Function{create-bit-vectors}{$k, \BWT, G, Q$}
      \State{compute the $\LCP$-array with the help of the $\BWT$}
      \State{compute the array $C$ of size $\sigma$}\label{C-array}
                  \State{initialize two bit vectors $\BVr$ and $\BVl$ of length $n$ with zeros}
      \State{$lb \gets 1, kIndex \gets 0, \mathit{lastdiff} \gets 0$, $open \gets $ false, $counter\gets 1$}
      \For{$i \gets 2$ \Keyw{to} $n+1$} \Comment{$\LCP[1] = \LCP[n+1] = -1$}\label{for-loop start}
          \State{increment $C[\BWT[i-1]]$ by one}\label{increment C-array}
          \If{$\LCP[i] \geq k$}\label{begin k-mers}
		\Let{$open$}{true}
		\If{$\LCP[i] = k$}
			\Let{$kIndex$}{$i$}
		\EndIf
          \Else \Comment{$\LCP[i] < k$}
		\If{$open$}
			\If{$kIndex > lb$}
				\Let{$\BVr[lb]$}{$1$}\label{set bitvector at lb}
				\Let{$\BVr[i-1]$}{$1$}\label{end k-mers}
                                \Let{$G[counter]$}{$(k,lb,i-lb,lb)$}\label{add node to G}
                                \State{$enqueue(Q,counter)$}\label{add node to Q}
                                \Let{$counter$}{$counter+1$}
			\EndIf
			\If{$\mathit{lastdiff} > lb$}\label{lastdiff}
			        \For{$j \gets lb$ \Keyw{to} $i-1$}
					\Let{$c$}{$\BWT[j]$}\label{bit vector 3 start}
                                        \If{$c \notin \{\#,\$\}$}\Comment{stop nodes will get different identifiers}
					     \Let{$\BVl[C[c]]$}{$1$}\label{bit vector 3 stop}\label{set bit in vector 3}
			                \EndIf
			        \EndFor
			\EndIf
			\Let{$open$}{false}
		\EndIf
		\Let{$lb$}{$i$}
          \EndIf
          \If{$\BWT[i] \neq \BWT[i-1]$}\label{set lastdiff start}
                \Let{$\mathit{lastdiff}$}{$i$}\label{set lastdiff stop}
          \EndIf
      \EndFor\label{for-loop stop}
      \Let{$open$}{false}\label{refine bit vector 3 start}
       \For{$i \gets 1$ \Keyw{to} $n+1$}
          \If{$open$}
		\Let{$\BVl[i]$}{0}
		\If{$\BVr[i] = $ $1$}
			\Let{$open$}{false}
		\EndIf
          \ElsIf{$\BVr[i] =$ $1$}
		\Let{$\BVl[i]$}{0}
		\Let{$open$}{true}
          \EndIf
      \EndFor\label{refine bit vector 3 stop}
      \State{\Return $(\BVr, \BVl)$}
      \EndFunction
  \end{algorithmic}
\end{algorithm}

We would like to stress that all right-maximal $k$-mers
can be determined without the entire $\LCP$-array. In order to verify whether
or not an interval satisfies properties (1)--(4), it is sufficient to compute
all entries $\leq k$ in the $\LCP$-array (the others have a value $>k$).
Since the algorithm of \cite{BEL:GOG:OHL:SCH:2013} calculates entries
in the $\LCP$-array in ascending order, it is ideal for our purposes.
We initialize an array $L$ with values $2$ and set $L[1]=0$ and $L[n+1]=0$.
Two bits are enough to encode the case ``$< k$'' by $0$, the case ``$= k$''
by $1$, and the case ``$> k$'' by $2$ (so initially all entries
in the $\LCP$-array are
marked as being $>k$, except for $L[1]$ and $L[n+1]$, which are marked as
being $<k$). Then, for $\ell$ from $0$ to $k-1$, the algorithm of
\cite{BEL:GOG:OHL:SCH:2013} calculates all indices $p$
with entries $\LCP[p] = \ell$ and sets $L[p] = 0$. Furthermore, it continues
to calculates all indices $q$ with entries $\LCP[q] = k$ and sets $L[q] = 1$.
Now the array $L$ contains all the information that is needed to compute
right-maximal $k$-mers.

As already mentioned, in pan-genome analysis
$S=S^1\#S^2\#\dots S^{\seqnumber-1}\#S^\seqnumber\$$
is the concatenation of multiple genomic
sequences $S^1,\dots,S^\seqnumber$, separated by a special symbol $\#$.
Our algorithm treats the different occurrences of
$\#$ as if they were different characters.
Assuming that $\#$ is the second smallest character, this can be
achieved as follows. As explained above,
all right-maximal $k$-mers can be determined without the
entire $\LCP$-array if the algorithm in \cite{BEL:GOG:OHL:SCH:2013} is used.
If there are $\seqnumber-1$ occurrences of $\#$ in total and this algorithm
starts with $\seqnumber-1$ singleton intervals $[s..s]$,
$2\leq s \leq \seqnumber$, instead of the
$\#$-interval $[2..\seqnumber]$, then the different occurrences of
$\#$ are treated as if they were different characters.

Bit vector $\BVl$ is computed on lines \ref{for-loop start} to
\ref{for-loop stop} of Algorithm \ref{alg:debruijnhelper}
as follows: If the suffix array interval $[lb..rb]$ of a repeat $\omega$ of
length $\geq k$ is detected, then it must be checked whether or
not $\omega$ is left-maximal (note that $\RightBoundary = i-1$).
Recall that $\omega$ is a left-maximal repeat if and only if
$|\{\BWT[\LeftBoundary],\BWT[\LeftBoundary+1],\dots, \BWT[\RightBoundary]\}| \geq 2$. Algorithm \ref{alg:debruijnhelper} checks this condition
by keeping track of the largest index $\mathit{lastdiff}$ at which the
characters $\BWT[\mathit{lastdiff}-1]$ and $\BWT[\mathit{lastdiff}]$ differ;
see lines \ref{set lastdiff start} and \ref{set lastdiff stop}.
Since $\mathit{lastdiff} \leq \RightBoundary = i-1$, the characters
$\BWT[\LeftBoundary],\BWT[\LeftBoundary+1],\dots, \BWT[\RightBoundary]$
are not all the same if and only if $\mathit{lastdiff} > \LeftBoundary$.
If this condition on line \ref{lastdiff} evaluates to true, then
for each $c\notin \{\#,\$\}$ in $\BWT[lb..rb]$ the algorithm
sets $\BVl[\LF[q]]$ to $1$ in lines
\ref{bit vector 3 start} to \ref{bit vector 3 stop}, where $q$ is the index
of the last occurrence of $c \in \BWT[lb..rb]$
and $\LF$ is the last-to-first mapping. How this is done by means of
the $C$-array will be explained below. So a one in $\BVl$ marks
a $k$-mer that precedes a left-maximal $k$-mer. Since we are only interested
in those $k$-mers that are not right-maximal (right-maximal $k$-mers
are already covered by bit vector $\BVr$), lines
\ref{refine bit vector 3 start} to \ref{refine bit vector 3 stop}
of Algorithm \ref{alg:debruijnhelper} reset those one-bits in $\BVl$ to zero
that mark a right-maximal $k$-mer.

It remains for us to explain the computation of the $\BVl$ vector with the
$C$-array. After the computation of the $C$-array on line \ref{C-array} of
Algorithm \ref{alg:debruijnhelper}, for each $c \in \Sigma$,
$C[c]$ is the overall number of occurrences of characters in $S$ that are
strictly smaller than $c$. Moreover, after line \ref{increment C-array} of
Algorithm \ref{alg:debruijnhelper} was executed, we have
$C[\BWT[i-1]]=\LF[i-1]$ (to see this, recall from
Section \ref{sec-Preliminaries} how the $\LF$-mapping can be
computed from the $\BWT$).
Thus, when the for-loop on lines \ref{for-loop start} to
\ref{for-loop stop} of Algorithm \ref{alg:debruijnhelper}
is executed for a certain value of $i$, we have $C[c]=\LF[q] $
for each character $c$ in $\BWT[1..i-1]$, where $q$ is the index
of the last occurrence of $c$ in $\BWT[1..i-1]$.
Algorithm \ref{alg:debruijnhelper} uses this fact on line
\ref{set bit in vector 3}: $C[c]=\LF[q] $, where $q$ is the index
of the last occurrence of $c$ in $\BWT[lb..i-1]$.

Apart from the direct construction of the $\LCP$-array
from the $\BWT$, which takes $O(n \log \sigma)$ time,
Algorithm \ref{alg:debruijnhelper} has a linear run-time.
The overall run-time is therefore $O(n \log \sigma)$.

\subsection{Construction of the space-efficient representation}

Algorithm \ref{alg-computation of a compressed de Bruijn graph}
constructs the implicit representation of the compressed de Bruijn
graph. It calls Algorithm \ref{alg:debruijnhelper}, which computes---besides
the two bit vectors $\BVr$ and $\BVl$---the suffix array interval
$[lb..lb+size-1]$ of each right-maximal $k$-mer $\omega$, stores
the quadruple $(k,lb,size,lb)$ at $G[id]$, where $id = (rank_1(\BVr,lb)+1)/2$
(this is because Algorithm \ref{alg:debruijnhelper} computes
right-maximal $k$-mer intervals in lexicographical order),
and adds $id$ to the (initially empty) queue $Q$. The attributes
$G[id].size$ and $G[id].\suffixlb$ will never change, but the attributes
$G[id].len$ and $G[id].lb$ will change when a left-extension is possible.

\begin{algorithm}[H]
  \caption{Construction of the implicit compressed de Bruijn graph.
    \label{alg-computation of a compressed de Bruijn graph}}
\begin{algorithmic}[1]
\Function{create-compressed-graph}{$k, \BWT$}
    \State{create an empty graph $G$}
    \State{create an empty queue $Q$}
    \Let{$(\BVr, \BVl)$}{\textsc{create-bit-vectors}($k, \BWT, G, Q$)}
    \Let{$rightMax$}{$rank_1(\BVr,n)/2$}
    \Let{$leftMax$}{$rank_1(\BVl,n)$}
    \For{$s \gets 1$ \Keyw{to} $\seqnumber$} \Comment{add the stop nodes for the $\seqnumber$ sequences} \label{begin for-loop}
         \Let{$id$}{$rightMax+leftMax+s$}
        \Let{$G[id]$}{$(1,s,1,s)$} \label{insert stopNode into G}
        \State{$enqueue(Q,id)$}
        \Let{$\BVl[s]$}{$0$}\label{end for-loop}

    \EndFor
    \While{$Q$ is not empty}\label{begin while-loop}
        \Let{$id$}{$dequeue(Q)$}
        \Repeat
            \Let{$extendable$}{false}
            \Let{$lb$}{$G[id].lb$}
            \Let{$rb$}{$lb + G[id].size - 1$}
            \Let{$list$}{$getIntervals([lb..rb])$}
                \For{\Keyw{each} $(c, [i..j])$ \Keyw{in} $list$}
                    \Let{$ones$}{$rank_1(\BVr,i)$}
                    \If{$ones$ is even \Keyw{and} $\BVr[i] = 0$}
                        \If{$c \notin \{\#,\SNT\}$}
                            \If{$list$ contains just one element}\Comment{Case 1}
                                \Let{$extendable$}{true}
                                \Let{$G[id].len$}{$G[id].len+1$}
                                \Let{$G[id].lb$}{$i$}
                            \Else \Comment{Case 2}\label{Case 2}
                                \Let{$newId$}{$rightMax + rank_1(\BVl,i-1)+1$}
                                \Let{$G[newId]$}{$(k,i,j-i+1,i)$}
                                \State{$enqueue(Q,newId)$}\label{enqueue node2}
                            \EndIf
                        \EndIf
                    \EndIf
                \EndFor
        \Until{\Keyw{not} $extendable$}
    \EndWhile        \label{end while-loop}
\EndFunction
\end{algorithmic}
\end{algorithm}

In the for-loop on lines \ref{begin for-loop}--\ref{end for-loop},
the stop nodes are added to $G$ and their identifiers are added to $Q$.
In the while-loop on lines \ref{begin while-loop}--\ref{end while-loop},
as long as the queue $Q$ is not empty,
the algorithm removes an identifier $id$ from $Q$ and in a repeat-loop
computes $list = getIntervals([lb..rb])$, where $lb=G[id].lb$ and
$rb= lb + G[id].size - 1$. During the repeat-loop,
the interval $[lb..rb]$ is the suffix array interval of some string
$\omega$ of length $G[id].len$. In the body of the repeat-loop, a flag
$extendable$ is set to false. The procedure call
$getIntervals([lb..rb])$ then returns the list $list$
of all $c\omega$-intervals. At this point, the algorithm tests whether or
not the length $k$ prefix of $c\omega$ is a right-maximal repeat.
It is not difficult to see that the length $k$ prefix of $c\omega$
is a right-maximal repeat if and only if the $c\omega$-interval $[i..j]$
is a subinterval of a right-maximal $k$-mer interval. Here,
the bit vector $\BVr$ comes into play. At the beginning of
Algorithm \ref{alg-computation of a compressed de Bruijn graph},
all suffix array intervals of right-maximal $k$-mers have been computed
and their boundaries have been marked in $\BVr$. It is crucial to note
that these intervals are disjoint. Lemma \ref{lem-correctness} shows
how the bit vector $\BVr$ can be used to test for non-right-maximality.

\begin{lemma}
\label{lem-correctness}
The $c\omega$-interval $[i..j]$
is not a subinterval of a right-maximal $k$-mer interval if and only if
$rank_1(\BVr,i)$, the number of ones in $\BVr$ up to (and including) position $i$,
is even and $\BVr[i] = 0$.
\end{lemma}
\begin{proof}
``only-if:'' Suppose $[i..j]$ is not a subinterval of a right-maximal $k$-mer
interval. Since $[i..j]$ cannot overlap with a right-maximal $k$-mer interval,
it follows that $rank_1(\BVr,i)$ must be even and $\BVr[i..j]$ contains
only zeros.\\
``if:'' Suppose $[i..j]$ is a subinterval of a right-maximal $k$-mer interval
$[p..q]$. If $i\neq j$, then $rank_1(\BVr,i)$ must be odd. If $i = j$, then
$rank_1(\BVr,i)$ may be even. But in this case $i$ must be the right boundary
of the interval $[p..q]$, so $\BVr[i] = \BVr[q] =1$.
\end{proof}

Now, the algorithm proceeds by case analysis.
If the length $k$ prefix of $c\omega$ is a right-maximal repeat,
there must be a node $v$ that ends with the length $k$
prefix of $c\omega$ (note that $c\omega[1..k]$ and $\omega$ have a
suffix-prefix-overlap of $k-1$ characters),
and this node $v$ will be detected by a computation
that starts with the $k$-mer $c\omega[1..k]$. Consequently, the
computation stops here.
If the length $k$ prefix of $c\omega$ is not a right-maximal repeat,
one of the following two cases occurs:
\begin{enumerate}
\item If $list$ contains just one element $(c, [i..j])$, then
$\omega$ is not left-maximal. In this case, the algorithms sets
$extendable$ to true, $G[id].lb$ to $i$, and increments
$G[id].len$ by one. Now $G[id]$ represents the $c\omega$-interval $[i..j]$
and the repeat-loop continues with this interval. Note that
$G[id].size = j-i+1$ because $\omega$ is not left-maximal.
\item Otherwise, $\omega$ is left-maximal.
In this case, a split occurs (so the attributes of $G[id]$ will not change
any more) and
Algorithm \ref{alg-computation of a compressed de Bruijn graph}
must continue with the $k$-mer prefix $x=c\omega[1..k]$ of $c\omega$.
For the correctness of the algorithm, it is
important to note that the interval $[i..j]$ is the
$x$-interval; see Lemma \ref{lem-correctness2}.
We use the bit vector $\BVl$ to assign the unique identifier
$newId=rightMax + rank_1(\BVl,i-1)+1$ to the next node, which
corresponds to (or ends with) $x$ (recall that $rightMax$ is the number of
all right-maximal $k$-mers and that $x$ is not a right-maximal $k$-mer).
So a quadruple $(k,i,j-i+1,i)$ is inserted at $G[newId]$
and $newId$ is added to $Q$.
\end{enumerate}

\begin{lemma}
\label{lem-correctness2}
Consider the $c\omega$-interval $[i..j]$ in Case 2 of
Algorithm \ref{alg-computation of a compressed de Bruijn graph}
(beginning at line \ref{Case 2}).
The interval $[i..j]$ coincides with the $c\omega[1..k]$-interval $[p..q]$.
\end{lemma}

\begin{proof}
Clearly, $[i..j]$ is a subinterval of $[p..q]$ because $c\omega[1..k]$ is
a prefix of $c\omega$. For a proof by contradiction, suppose that
$[i..j] \neq [p..q]$. Let $cu$ be the longest common prefix of all suffixes
in the interval $[p..q]$. Note that the length $\ell$ of $cu$ is at least $k$.
Since $[i..j] \neq [p..q]$, it follows that there must be a suffix in the
interval $[p..q]$ that has a prefix $cub$ so that $cu$ is a proper prefix of
$c\omega$ and $b\neq c\omega[\ell+1]$. Consequently,
$cu$ is a right-maximal repeat. Clearly, this implies that
$u$ is a right-maximal repeat as well. We consider two cases:
\begin{enumerate}
\item $\ell=k$: In this case,
Algorithm \ref{alg-computation of a compressed de Bruijn graph}
stops (the length $k$ prefix $cu$ of $c\omega$ is a right-maximal
repeat), so it cannot execute Case 2; a contradiction.
\item $\ell > k$: Note that $u$ has length $\ell-1 \geq k$.
Since $u$ is a right-maximal repeat, it is impossible that
the procedure $getIntervals$ is applied to the $\omega$-interval $[lb..rb]$.
This contradiction proves the lemma.
\end{enumerate}
\end{proof}

As an example, we apply
Algorithm \ref{alg-computation of a compressed de Bruijn graph}
to $k=3$ and the $\LCP$-array and the $\BWT$ of the string
\texttt{ACTACGTACGTACG\$}; see Fig.\ \ref{fig:suffix array}.
There is only one right maximal $k$-mer, \texttt{ACG}, so a node
$(len,lb,size,\suffixlb) = (3,2,3,2)$ is inserted at $G[1]$
and the identifier $1$ is added to the queue $Q$ in
Algorithm \ref{alg:debruijnhelper}.
On line \ref{insert stopNode into G} of
Algorithm \ref{alg-computation of a compressed de Bruijn graph}
the stop node is added to $G$. It has the identifier
$rightMax+leftMax+1 = 1+2+1=4$, so $G[4]$ is set to $(1,1,1,1)$
and $4$ is added to $Q$. In the while-loop, the identifier $1$
of node $(3,2,3,2)$ is dequeued and the procedure call
$getIntervals([2..4])$ returns a list that contains just one interval,
the \texttt{TACG}-interval $[13..15]$. Since $rank_1(\BVr,13) = 2$ is even and
$\BVr[13]=0$, Case 1 applies. So $extendable$ is set to true
and $G[1]$ is modified to $(4,13,3,2)$.
In the next iteration of the repeat-loop, $getIntervals([13..15])$
returns the list $[(\texttt{C},[9..9]), (\texttt{G},[11..12])]$,
where $[9..9]$ is the \texttt{CTACG}-interval and $[11..12]$ is the
\texttt{GTACG}-interval. It is readily verified that Case 2 applies in both
cases. For the \texttt{CTACG}-interval $[9..9]$ we obtain the identifier
$rightMax+rank_1(\BVl,9-1)+1 = 1+0+1=2$, so $G[2]$ is set to $(3,9,1,9)$.
Analogously, the \texttt{GTACG}-interval $[11..12]$ gets the identifier
$rightMax+rank_1(\BVl,11-1)+1 = 1+1+1=3$ and $G[3]$ is set to $(3,11,2,11)$.
Furthermore, the identifiers $2$ and $3$ are added to the queue $Q$. Next, the
identifier $4$ of the stop node $(1,1,1,1)$ is dequeued and the procedure call
$getIntervals([1..1])$ returns a list that contains just one interval,
the \texttt{G\$}-interval $[10..10]$. Case 1 applies, so $G[4]$ is
modified to $(2,10,1,1)$. In the second iteration of the repeat-loop,
$getIntervals([10..10])$ returns the \texttt{CG\$}-interval $[6..6]$.
Again Case 1 applies and $G[4]$ is modified to $(3,6,1,1)$.
In the third iteration of the repeat-loop,
$getIntervals([6..6])$ returns the \texttt{ACG\$}-interval $[2..2]$.
This time, $rank_1(\BVr,2)=1$ is odd and therefore the repeat-loop
terminates. The computation
continues until the queue $Q$ is empty; the final compressed
de Bruijn graph is shown in Fig.\ \ref{fig:space-efficient representation}.

We claim that Algorithm \ref{alg-computation of a compressed de Bruijn graph}
has a worst-case time complexity of $O(n \log \sigma)$
and use an amortized analysis to prove this. Since the
compressed de Bruijn graph has at most $n$ nodes, it is
an immediate consequence that at most $n$ identifiers
enter and leave the queue $Q$ (this covers Case 2).
Case 1 can occur at most $n$ times because there are
at most $n$ left-extensions; so at most $n$ intervals
generated by the procedure $getIntervals$ belong to this category.
Each left-extension eventually ends; so at most $n$ intervals
generated by the procedure $getIntervals$ belong to this category
because there are at most $n$ left-extensions.
In summary, at most $2n$ intervals are
generated by the procedure $getIntervals$. Since this procedure takes
$O(\log \sigma)$ time for each generated interval, the claim follows.

\subsection{Construction of the explicit compressed de Bruijn graph}
\label{sec-Construction of the explicit compressed de Bruijn graph}

In this section, we show how the explicit compressed de Bruijn graph
can be constructed efficiently from the implicit representation.
If the pan-genome consists of $\seqnumber$ sequences, then
$S=S^1\#S^2\#\dots S^{\seqnumber-1}\#S^\seqnumber\$$ and there are $d$ stop
nodes. Since the implicit representation allows for an efficient backward
traversal, there is no need for start nodes. By contrast, the explicit graph
must provide them. That is why
Algorithm \ref{alg-construct the explicit compressed de Bruijn graph}
stores them in an array $StartNodes$ of size $d$.

Algorithm \ref{alg-construct the explicit compressed de Bruijn graph}
starts with the stop node of the last sequence $S^\seqnumber$,
which has identifier $id=rightMax+leftMax+1$. Let $\omega$ be the string
corresponding to node $id$. Since $\omega$ ends with $\$$ and $\$$
appears at position $n$ in $S$, the start position of $\omega$ in $S$
is $pos=n-G[id].len +1$. Consequently, $pos$ is added to the front of
$G[id].posList$ on line \ref{add pos to stop node} of
Algorithm \ref{alg-construct the explicit compressed de Bruijn graph}.
Next, we have to find the predecessor of node $id$. It is not
difficult to see that $idx=G[id].lb$ is the index in the suffix array
at which the suffix $S_{pos}$ can be found (note that $S_{pos}$ has
$\omega$ as a prefix). Clearly, $i=LF(idx)$ is the index of
the suffix $S_{pos-1}$ in the suffix array. Note that $S_{pos-1}$ has
$c\omega$ as a prefix, where $c=\BWT[idx]$. If $c$ is not a separator
symbol (i.e., $c \notin \{\#,\SNT\}$), then the predecessor of node
$id$ is the node  $newId$ whose corresponding string $u$
ends with the $k$-mer prefix  $x = c\omega[1..k]$ of $c\omega$.
If $x$ is a right-maximal $k$-mer, then $newId$ is
$(rank_1(\BVr,i)+1)/2$, otherwise it is $rightMax + rank_1(\BVl,i-1)+1$.
Note that $u$ ends at position $pos-1+(k-1)$ in $S$ because $u$ and $\omega$
overlap $k-1$ characters. It follows as a consequence that $u$ starts at
position $newPos = pos-1+k-1-G[newId].len +1 = pos-1-(G[newId].len - k)$.
So the position $newPos$ is added to the front of the position list
of $G[newId]$. Because node $G[id]$ is the successor of node $G[newId]$,
the identifier $id$ is added to the front of the adjacency list $G[newId]$.
To find the predecessor of node $newId$ in the same fashion, we must find the
index $idx$ at which the suffix $S_{newPos}$ can be found in the suffix array.
According to Lemma \ref{lem-correctness3}, this is
$G[newId].lb+(i-G[newId].\suffixlb)$. The while-loop repeats the
search for a predecessor node until a separator symbol is found.
In this case, a start node has been reached and its identifier
is stored in an array $StartNodes$ of size $d$.
Since there are $d$ separator symbols, the whole process is executed $d$ times.

\begin{lemma}
\label{lem-correctness3}
Let $G[id] = (len,lb,size,\suffixlb)$ be a node in the implicit representation
of the compressed de Bruijn graph. If $G[id]$ is not a stop node and
suffix $S_p$ appears at index $i$ in the
interval $[b..e]=[\suffixlb..\suffixlb+size-1]$ (i.e., $\SA[i]=p$), then
the suffix $S_{p+(len-k)}$ appears at index $lb+(i-\suffixlb)$ in the
interval $[lb..lb+size-1]$.
\end{lemma}

\begin{proof}
Let $u$ be the string corresponding to $G[id]$ and let $x$ be the
$k$-mer suffix of $u$. By construction, $[lb..lb+size-1]$ is the
$u$-interval and $[b..e]$ is the $x$-interval in the suffix array.
If $u=x$, then $len=k$, $lb = \suffixlb$, and there is nothing to
show. So suppose $u\neq x$ and let $c$ be the character that
precedes $x$ in $u$ (recall that $x$ is not left-maximal).
Since $S_{\SA[b]} < S_{\SA[b+1]} < \dots < S_{\SA[e]}$,
it follows that $cS_{\SA[b]} < cS_{\SA[b+1]} < \dots < cS_{\SA[e]}$.
In other words, the $cx$-interval contains the suffixes
$S_{\SA[b]-1} < S_{\SA[b+1]-1} < \dots < S_{\SA[e]-1}$.
Consequently, if $i$ is the $q$-th element of $[b..e]$ and $\SA[i]=p$,
then $\LF(i)$ is the $q$-th element of the $cx$-interval and $\SA[\LF(i)]=p-1$
(this implies in particular that $[\LF(b)..\LF(e)]$ is the $cx$-interval).
Iterating this argument $len-k$ times yields the lemma.
\end{proof}

\begin{algorithm}[H]
  \caption{Construction of the explicit compressed de Bruijn graph.
    \label{alg-construct the explicit compressed de Bruijn graph}}
\begin{algorithmic}[1]
\Function{construct-explicit-graph}{$G,\BWT,\LF,\BVr, \BVl$}
    \Let{$i$}{$1$}
    \Let{$pos$}{$n$}
    \For{$s \gets 1$ \Keyw{to} $\seqnumber$}\Comment{there are $d$ occurrences of $\#$ and $\SNT$ in $S$}
        \Let{$id$}{$rightMax+leftMax+i$}
        \Let{$pos$}{$pos-G[id].len +1$}
        \State{add $pos$ to the front of $G[id].posList$}\label{add pos to stop node}
        \Let{$idx$}{$G[id].lb$}
        \While{$\BWT[idx] \notin \{\#,\SNT\}$}
            \Let{$i$}{$LF(idx)$}
            \Let{$ones$}{$rank_1(\BVr,i)$}
            \If{$ones$ is even \Keyw{and} $\BVr[i] = 0$}
                \Let{$newId$}{$rightMax + rank_1(\BVl,i-1)+1$}
            \Else
                \Let{$newId$}{$\lfloor(ones+1)/2\rfloor$}
            \EndIf
            \Let{$newPos$}{$pos-1-(G[newId].len -k)$}
            \State{add $newPos$ to the front of $G[newId].posList$}
            \State{add $id$ to the front of $G[newId].adjList$}
            \Let{$idx$}{$G[newId].lb+(i-G[newId].\suffixlb)$}
            \Let{$id$}{$newId$}
            \Let{$pos$}{$newPos$}
        \EndWhile
        \Let{$StartNodes[d+1-s]$}{$id$}
        \Let{$i$}{$LF(idx)$}
    \EndFor
\EndFunction
\end{algorithmic}
\end{algorithm}

Algorithm \ref{alg-construct the explicit compressed de Bruijn graph}
has a worst-case time complexity of $O(N\log \sigma)$, where
$N$ is the number of edges in the compressed de Bruijn graph.
This is because in each execution of the while-loop an edge is added
to the graph and a value $\LF(idx)$ is computed in $O(\log \sigma)$ time
(all other operations take only constant time). Since the uncompressed
de Bruijn graph has at most $n$ edges, so does the compressed graph.
Hence $N \leq n$. In fact, $N$ is much smaller than $n$ in virtually all cases.
It follows from the preceding section that $N$ can be characterized in terms
of left- and right-maximal $k$-mer repeats. We have seen that the number
of nodes in the compressed de Bruijn graph
equals $|V_1|+|V_2|+d=rightMax+leftMax+d$, where
$V_1 = \{\omega \mid \omega \mbox{ is a right-maximal $k$-mer repeat in } S\}$
and
$V_2 = \{\omega \mid \exists i \in \{1,\dots,n-k\} : \omega = S[i..i+k-1] \notin V_1$ and $S[i+1..i+k] \mbox{ is a left-maximal $k$-mer repeat in } S\}$;
the stop nodes are taken into account by adding $d$.
The number $N$ of edges in the compressed de Bruijn graph therefore is
$|\{i \mid 1 \leq i \leq n-k \mbox{ and } S[i..i+k-1] \in V_1 \cup V_2\}|$.

\section{Operations on the compressed de Bruijn graph}
\label{sec:Operations on the compressed de Bruijn graph}

It is our next goal to show how the combination of the implicit graph
and the FM-index can be used to search for a pattern $P$ of length $m\geq k$.
This is important, for example, if one wants to search for
a certain allele in the pan-genome and---if it is present---to examine
the neighborhood of that allele in the graph.
Algorithm \ref{alg-find path in the compressed de Bruijn graph}
shows pseudo-code for such a search. The main difficulty is to find
the node of the $k$-length suffix of $P$ in the implicit graph. Once we
have found this node, we can use the method introduced in the previous
section to continue the search (where backward search replaces the
$\LF$-mapping).

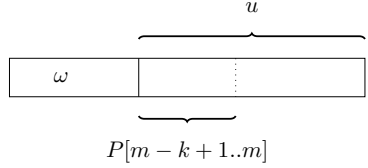
\begin{figure}[H]
\begin{center}
\scalebox{0.85}{
\begin{tikzpicture}[
x=1.00cm, y=1.00cm,
]

\coordinate (a_u) at ( 0mm, 6mm);
\node             at ( 8mm, 3mm) {$\omega$};
\coordinate (b_u) at (20mm, 6mm);
\coordinate (c_u) at (35mm, 6mm);
\coordinate (d_u) at (55mm, 6mm);
\coordinate (a_d) at ([yshift=-6mm]a_u);
\coordinate (b_d) at ([yshift=-6mm]b_u);
\coordinate (c_d) at ([yshift=-6mm]c_u);
\coordinate (d_d) at ([yshift=-6mm]d_u);

\draw         (a_u) rectangle (d_d);
\draw         (b_u) to (b_d);
\draw[dotted] (c_u) to (c_d);
\draw[thick, decoration={brace, mirror, raise=3mm}, decorate] (b_d) to node [anchor=north, yshift=-5.5mm] {$P[m-k+1..m]$} (c_d);
\draw[thick, decoration={brace, mirror, raise=3mm}, decorate] (d_u) to node [anchor=north, yshift=+10mm] {$u$} (b_u);

\end{tikzpicture}
}
\caption{The string $\omega$ has $u$ as suffix and $u$
has $P[m-k+1..m]$ as prefix.
\label{fig-pattern}}
\end{center}
\end{figure}

Using the FM-index, we first find the suffix array interval $[i..j]$ of
the $k$-mer suffix $P[m-k+1..m]$ of $P$. If $i\leq j$ (i.e.,
$P[m-k+1..m]$ occurs in the pan-genome), we search for the node $G[id]$
whose corresponding string $\omega$ contains $P[m-k+1..m]$.
If $P[m-k+1..m]$ is a suffix of $\omega$, then the unknown identifier
$id$ can be determined by lines \ref{begin find id}--\ref{end find id}
in Algorithm \ref{alg-find path in the compressed de Bruijn graph}.
If it is not a suffix of $\omega$, then there is a suffix $u$
of $\omega$ that has $P[m-k+1..m]$ as prefix; see Fig.\ \ref{fig-pattern}.
The key observation is that $[i..j]$ is the suffix array interval of $u$.
Moreover, $u$ can be written as $c_1c_2\dots c_{\ell} x$, where $c_q \in \Sigma$
for $q\in \{1,\dots,\ell\}$ and $x$ is the $k$-mer suffix of $u$.
Note that the value of $\ell$ is unknown.
Since $c_2\dots c_{\ell} x$ is not left-maximal, it follows that
$[\Psi(i)..\Psi(j)]$ is its suffix array interval (this can be proven
by similar arguments as in the proof of Lemma \ref{lem-correctness3}).
Algorithm \ref{alg-find path in the compressed de Bruijn graph}
iterates this process until either on line \ref{stop node id} the identifier
of a stop node or on lines \ref{begin find id}--\ref{end find id}
the identifier of a non-stop-node is found. In the latter case,
there are $\ell$ characters before the $k$-mer suffix $x$ of $u$;
so $|u|=\ell +k$ and therefore $G[id].len - \ell -k$ characters precede $u$ in
$\omega$ (see line \ref{start of u}). In the former case,
$u=c_1c_2\dots c_{\ell} \#$ has length $\ell +1$ and thus
$G[id].len - \ell -1$ characters precede $u$ in $\omega$.
To obtain this value on line \ref{start of u}, $k$ is subtracted from
$\ell +1$ on line \ref{subtract k}.

\begin{algorithm}[H]
  \caption{Find the path of a pattern $P$ with $|P| = m$ in the compressed de Bruijn graph.
    \label{alg-find path in the compressed de Bruijn graph}}
\begin{algorithmic}[1]
\Function{find-nodes}{$P$}
    \Let{$[i..j]$}{$backwardSearch(P[m-k+1..m])$} \Comment{$k$-length suffix of $P$}
    \If{$i>j$}
        \State{\Keyw{return} an empty list} \Comment{$k$-length suffix of $P$ does not occur in the input}
    \EndIf
    \Let{$[lb..rb]$}{$[i..j]$}
    \Let{$id$}{$\bot$}
    \Let{$\ell$}{$0$}
    \While{$id=\bot$} \Comment{search for the node that contains the suffix of length $k$}
        \Let{$ones$}{$rank_1(\BVr,i)$}\label{begin find id}
        \If{$ones$ is odd \Keyw{or} $\BVr[i] = 1$}
            \Let{$id$}{$\lfloor(ones+1)/2\rfloor$}
        \ElsIf{$\BVl[i..j]$ contains a $1$}
            \Let{$id$}{$rightMax + rank_1(\BVl,i-1)+1$}\label{end find id}
        \Else \label{else statement}
            \Let{$i$}{$\Psi(i)$}\Comment{$\Psi$ is the inverse of $\LF$}
            \Let{$j$}{$\Psi(j)$}
            \Let{$\ell$}{$\ell+1$}
            \If{$i \leq \seqnumber$} \Comment{stop node} \label{stop node id}
                \Let{$id$}{$rightMax + leftMax+i$}
                \Let{$\ell$}{$\ell+1-k$}    \label{subtract k}
            \EndIf

        \EndIf
    \EndWhile
    \Let{$\ell$}{$G[id].len - \ell - k$}\label{start of u}
    \Let{$resList$}{$[id]$}\label{resList}\Comment{a list containing only $id$}
    \Let{$[i..j]$}{$[lb..rb]$} \Comment{continue backwardSearch}
    \Let{$pos$}{$m-k$}
    \While{$i \leq j$ \Keyw{and} $pos > 0$}\label{begin while}
        \Let{$[i..j]$}{$backwardSearch(P[pos], [i..j])$}
        \Let{$pos$}{$pos-1$}
        \If{$\ell > 0$}
            \Let{$\ell$}{$\ell-1$}
        \Else
            \Let{$ones$}{$rank_1(\BVr,i)$}
            \If{$ones$ is even \Keyw{and} $\BVr[i] = 0$}
                \Let{$id$}{$rightMax + rank_1(\BVl,i-1)+1$}
            \Else
                \Let{$id$}{$\lfloor(ones+1)/2\rfloor$}
            \EndIf
            \State{add $id$ to the front of $resList$}
            \Let{$\ell$}{$G[id].len-k$}
        \EndIf
    \EndWhile\label{end while}
    \If{$i>j$}
        \State{\Keyw{return} an empty list} \Comment{$P$ does not occur in the input}
    \Else
        \State{\Keyw{return} $resList$}
    \EndIf
\EndFunction
\end{algorithmic}
\end{algorithm}

To summarize, after $\ell$ is set to its new value on line \ref{start of u}
of Algorithm \ref{alg-find path in the compressed de Bruijn graph}, we know
that $id$ is the identifier of the node whose corresponding string $\omega$
contains $P[m-k+1..m]$ and that there are $\ell$ characters preceding
$P[m-k+1..m]$ in $\omega$. On line \ref{resList} the list $resList$,
which will eventually contain the nodes corresponding to pattern $P$,
is initialized with the element $id$. In the while-loop on lines
\ref{begin while}--\ref{end while}, the backward search continues
with the character $P[pos]$ (where $pos=m-k$) and
the $P[m-k+1..m]$-interval $[i..j]$.
As long as $i\leq j$ (i.e., the suffix $P[pos+1..m]$ occurs in the pan-genome)
and $pos > 0$, $backwardSearch(P[pos], [i..j])$ yields the suffix array
interval of $P[pos..m]$ and $pos$ is decremented by one.
Within the while-loop there is a case distinction:
\begin{enumerate}
\item If $\ell > 0$, then the current prefix of $P[pos..m]$ still
belongs to the current node. In this case $\ell$ is decremented by one.
\item
If $\ell = 0$, then the $k$-mer prefix of $P[pos..m]$ belongs to the
predecessor node of the current node. Its identifier $id$ is determined in
the usual way and then added to the front of $resList$. The variable
$\ell$ is set to the new value $G[id].len-k$ because so many characters
precede the $k$-mer prefix of $P[pos..m]$ in the string corresponding
to node $G[id]$.
\end{enumerate}

Algorithm \ref{alg-find path in the compressed de Bruijn graph}
has a worst-case time complexity of $O((m+\ell) \log \sigma)$, where
$m = |P|$ and $\ell$ is the number of executions of the else-statement on line
\ref{else statement}. This is because the overall number of backward search
steps (each of which takes $O(\log \sigma)$ time) is $m$ and the number
of computations of $\Psi$-values (each of which also takes $O(\log \sigma)$
time) is $2\ell$. Of course, $\ell$ is bounded by the length of the
longest string corresponding to a node, but this can be proportional to $n$.
As a matter of fact, the worst case occurs when the algorithm gets a
de Bruijn sequence of order $k$ on the alphabet $\Sigma$ as input:
this is a cyclic string of length $n=\sigma^k$ containing every length $k$
string over $\Sigma$ exactly once as a substring. For example, the string
$aacagatccgctggtt$ is a de Bruijn sequence of order $k=2$
on the alphabet $\Sigma = \{a,c,g,t\}$. The compressed de Bruijn graph
for such a sequence has just one node and the corresponding string
is the de Bruijn sequence itself. In practice, however, $\ell$ is
rather small; see end of Section \ref{sec-experimental results}.

Algorithm \ref{alg-find path in the compressed de Bruijn graph}
finds the nodes in the compressed de Bruijn graph that
correspond to a pattern $P$. In this context, the following
(and similar) questions arise:
\begin{itemize}
\item In which sequences (or genomes) does pattern $P$ (or node $v$) occur?
\item In how many sequences (or genomes) does pattern $P$ (or node $v$) occur?
\item How often does pattern $P$ (or node $v$) occur in a specific sequence
(or genome)?
\end{itemize}
To answer these questions efficiently, we employ the document array
$D$ of size $n=|S|$. An entry $D[i] =j$ means that the suffix $S_{\SA[i]}$
belongs to (or starts within) the sequence $S^j$, where
$j\in \{1,\dots,\seqnumber\}$. The document array can be constructed
in linear time from the suffix array or the $\BWT$; see e.g.\
\cite[p.\ 347]{OHL:2013}. If we store the document array in a wavelet tree,
then the above-mentioned questions can be answered as follows:
Given the suffix array interval $[lb..rb]$ of pattern $P$ (or node $v$),
the procedure call $getIntervals([lb..rb])$ on the wavelet tree of the
document array returns a list consisting of all sequence numbers
$j$ in which $P$ occurs plus the number of occurrences of $P$ in $S^j$.
The worst-case time complexity of the procedure $getIntervals$ is
$O(z + z \log (d/z))$, where $z$ is the number of elements in
the output list; see Section \ref{sec:Introduction}.

\section{Experimental results}
\label{sec-experimental results}

The experiments were conducted on a 64 bit Ubuntu 14.04.1 LTS (Kernel 3.13)
system  equipped with two ten-core Intel Xeon processors E5-2680v2 with 2.8 GHz
and 128GB of RAM (but no parallelism was used).
All programs were compiled with g++ (version 4.8.2) using the provided makefile.
As test files we used the \emph{E.coli} genomes listed in the supplementary
material of \cite{MAR:LEE:SCH:2014}.
Additionally, we used 5 different assemblies of the human reference genome
(UCSC Genome Browser assembly IDs: hg16, hg17, hg18, hg19, and hg38)
as well as the maternal and paternal haplotype of individual NA12878
(Utah female) of the 1000 Genomes Project; see \cite{ROZ:ABY:WAN:2011}.
Our software and test data are available at
\url{https://www.uni-ulm.de/in/theo/research/seqana.html};
splitMEM can be obtained from
\url{http://sourceforge.net/projects/splitmem/}.

We implemented the three algorithms $\algoA$--$\algoC$ described in the
preliminary version of this article \cite{BEL:OHL:2015} and our new
algorithm $\algoD$ using Simon Gog's library \textsf{sdsl}
\cite{GOG:BEL:MOF:PET:2014}.
Both $\algoA$ and $\algoB$ require at least $n \log n$ bits because the
suffix array must be kept in main memory. Hence Yuta Mori's fast algorithm
\textsf{divsufsort} can be used to construct the suffix array without
increasing the memory requirements. By contrast, $\algoC$ and $\algoD$
use a variant of the semi-external
algorithm described in \cite{BEL:ZWE:GOG:OHL:2013} to construct the $\BWT$.
Both $\algoC$ and $\algoD$ store the $\BWT$ in a wavelet tree and use
additional bit vectors; see
Section \ref{sec-Computation of right-maximal k-mers}.
The variants of the algorithms that appear in Table \ref{tab-results-runtime}
are as follows: $\algoCone$ and $\algoDone$ compress only
the additional bit vectors, $\algoCtwo$ and $\algoDtwo$ also
compress the (bit vectors in the) wavelet tree, whereas
$\algoC$ and $\algoD$ do not use these compression options at all.
In contrast to the other algorithms, $\algoD$ (and its variants)
constructs the implicit graph (instead of the explicit graph) and the
wavelet tree of the document array. For a comparison with the
other algorithms, we also measured (called $\algoDplus$) the construction
of the implicit and the explicit graph (i.e., the combination
of Algorithms \ref{alg-computation of a compressed de Bruijn graph}
and \ref{alg-construct the explicit compressed de Bruijn graph}).

The first part of Table \ref{tab-results-runtime} (in which the $k$ column
has the entries init) shows how much time (in seconds) an algorithm needs
to construct the index data structure and its maximum main memory usage in
bytes per base pair. In the experiments,
we built compressed de Bruijn graphs for the $62$ \emph{E.\ coli} genomes
(containing 310 million base pairs) using the $k$-mer lengths $50$,
$100$, and $500$. Table \ref{tab-results-runtime} shows the results of these
experiments. The run-times include the construction of the index,
but similar to splitMEM it is unnecessary to rebuild the index for
a fixed data set and varying values of $k$. The peak memory usage
reported in Table \ref{tab-results-runtime} includes the size of the index
\emph{and} the size of the compressed de Bruijn graph. Due to its large
memory requirements, splitMEM was not able to build a compressed
de Bruijn graph for all $62$ strains of \emph{E.\ coli}
on our machine equipped with 128\;GB of RAM. That is why we included
a comparison based on the first $40$ \emph{E.\ coli} genomes
(containing 199 million base pairs) of the data set.

\begin{table}[H]
\begin{center}
\setlength{\tabcolsep}{2mm}
\begin{footnotesize}
\begin{tabular}{rlrrrr}
\hline
       $k$ &            algorithm &        40 Ecoli &        62 Ecoli &          7xChr1 &            7xHG \\
\hline \\[-0.5em]
      init &            \splitmem &    117     (315.25) &    141     (317.00) &                   - &                   - \\
      init &       \algoA, \algoB &     38 \hd\hd(5.00) &     64 \hd\hd(5.00) &    380     (5.00) &                   - \\
      init &       \algoC, \algoD &    131 \hd\hd(1.32) &    202 \hd\hd(1.24) &  1,168     (1.24) & 20,341     (1.24) \\
[0.5em] \hline \\[-0.5em]
        50 &            \splitmem &  2,261     (572.19) &                   - &                   - &                   - \\
        50 &               \algoA &     57 \hd\hd(5.22) &     92 \hd\hd(5.34) &    596     (6.20) &                   - \\
        50 &               \algoB &     61 \hd\hd(8.49) &     97 \hd\hd(8.78) &    619     (9.98) &                   - \\
        50 &               \algoC &    188 \hd\hd(2.23) &    300 \hd\hd(2.26) &  1,733     (3.07) & 29,816     (2.77) \\
        50 &            \algoCone &    208 \hd\hd(1.81) &    346 \hd\hd(1.85) &  1,880     (2.66) & 31,472     (2.36) \\
        50 &            \algoCtwo &    236 \hd\hd(1.63) &    374 \hd\hd(1.66) &  2,318     (2.51) & 39,366     (2.22) \\
        50 &               \algoD &    164 \hd\hd(1.75) &    254 \hd\hd(1.82) &  1,419     (1.28) & 25,574     (1.96) \\
        50 &            \algoDone &    167 \hd\hd(1.46) &    257 \hd\hd(1.53) &  1,435     (1.28) & 25,866     (1.66) \\
        50 &            \algoDtwo &    179 \hd\hd(1.32) &    272 \hd\hd(1.24) &  1,526     (1.24) & 27,365     (1.39) \\
        50 &           \algoDplus &    172 \hd\hd(3.26) &    268 \hd\hd(3.35) &  1,515     (3.59) & 27,619     (3.88) \\
        50 &        \algoDoneplus &    176 \hd\hd(2.97) &    271 \hd\hd(3.06) &  1,541     (3.31) & 28,044     (3.64) \\
        50 &        \algoDtwoplus &    188 \hd\hd(2.66) &    289 \hd\hd(2.74) &  1,629     (2.96) & 29,517     (3.38) \\
[0.5em] \hline \\[-0.5em]
       100 &            \splitmem &  2,568     (572.20) &                   - &                   - &                   - \\
       100 &               \algoA &     59 \hd\hd(5.00) &     95 \hd\hd(5.00) &    595     (5.95) &                   - \\
       100 &               \algoB &     62 \hd\hd(7.89) &     99 \hd\hd(8.19) &    605     (9.74) &                   - \\
       100 &               \algoC &    188 \hd\hd(1.63) &    299 \hd\hd(1.68) &  1,738     (2.74) & 27,815     (2.23) \\
       100 &            \algoCone &    205 \hd\hd(1.50) &    326 \hd\hd(1.49) &  1,839     (2.33) & 30,401     (1.80) \\
       100 &            \algoCtwo &    232 \hd\hd(1.32) &    411 \hd\hd(1.29) &  2,340     (2.14) & 38,134     (1.66) \\
       100 &               \algoD &    174 \hd\hd(1.71) &    261 \hd\hd(1.79) &  1,422     (1.28) & 25,723     (1.94) \\
       100 &            \algoDone &    171 \hd\hd(1.42) &    264 \hd\hd(1.50) &  1,439     (1.28) & 26,040     (1.64) \\
       100 &            \algoDtwo &    185 \hd\hd(1.32) &    289 \hd\hd(1.24) &  1,544     (1.24) & 27,464     (1.37) \\
       100 &           \algoDplus &    178 \hd\hd(2.61) &    270 \hd\hd(2.73) &  1,486     (3.21) & 26,878     (3.36) \\
       100 &        \algoDoneplus &    175 \hd\hd(2.32) &    273 \hd\hd(2.44) &  1,500     (2.92) & 26,999     (3.07) \\
       100 &        \algoDtwoplus &    190 \hd\hd(2.01) &    299 \hd\hd(2.12) &  1,624     (2.68) & 28,665     (2.80) \\
[0.5em] \hline \\[-0.5em]
       500 &            \splitmem &  2,116     (570.84) &                   - &                   - &                   - \\
       500 &               \algoA &     72 \hd\hd(5.00) &    113 \hd\hd(5.00) &    620     (5.83) &                   - \\
       500 &               \algoB &     83 \hd\hd(7.17) &    117 \hd\hd(7.43) &    640     (9.66) &                   - \\
       500 &               \algoC &    194 \hd\hd(1.50) &    304 \hd\hd(1.49) &  1,752     (2.67) & 28,548     (2.07) \\
       500 &            \algoCone &    216 \hd\hd(1.50) &    325 \hd\hd(1.49) &  1,839     (2.19) & 30,488     (1.65) \\
       500 &            \algoCtwo &    241 \hd\hd(1.32) &    378 \hd\hd(1.29) &  2,319     (2.06) & 36,993     (1.50) \\
       500 &               \algoD &    184 \hd\hd(1.65) &    283 \hd\hd(1.74) &  1,453     (1.28) & 26,362     (1.93) \\
       500 &            \algoDone &    197 \hd\hd(1.35) &    287 \hd\hd(1.44) &  1,477     (1.28) & 26,545     (1.63) \\
       500 &            \algoDtwo &    213 \hd\hd(1.32) &    322 \hd\hd(1.24) &  1,622     (1.24) & 28,501     (1.36) \\
       500 &           \algoDplus &    185 \hd\hd(1.81) &    285 \hd\hd(1.90) &  1,509     (3.14) & 27,285     (3.14) \\
       500 &        \algoDoneplus &    198 \hd\hd(1.52) &    288 \hd\hd(1.61) &  1,535     (2.83) & 27,417     (2.79) \\
       500 &        \algoDtwoplus &    214 \hd\hd(1.32) &    323 \hd\hd(1.29) &  1,694     (2.56) & 29,283     (2.58) \\
\end{tabular}
\caption{The first column shows the $k$-mer size (an entry init means that only
the index data structure is constructed) and the second column specifies
the algorithm used in the experiment. The remaining columns show the run-times
in seconds and, in parentheses, the maximum main memory usage in bytes per
base pair (including the construction) for the data sets described in the text.
A minus indicates that the respective algorithm was not able to solve
its task on our machine equipped with $128$ GB of RAM.}
\label{tab-results-runtime}
\end{footnotesize}
\end{center}
\end{table}

The experimental results show that our algorithms are more than an order of
magnitude faster than splitMEM while using significantly less space
(two orders of magnitude). To show the scalability of the new algorithms,
we applied them to different assemblies of the human genome
(consisting of 23 chromosomes: the 22 autosomes and the X-chromosome).
The compressed de Bruijn graphs of their first chromosomes (7xChr1, containing
1,736 million base pairs) and the complete seven genomes (7xHG, containing
21,201 million base pairs) were built for the $k$-mer lengths $50$,
$100$, and $500$. One can see from Table \ref{tab-results-runtime}
that algorithms $\algoA$ and $\algoB$ are very fast, but 128\;GB of RAM
was not enough for them to successfully build the graph for the seven human
genomes (note that at least $5$ bytes per base pair are required).
So let us compare algorithms $\algoC$ and $\algoD$ (and their variants).
The construction of the explicit graph with $\algoDplus$ is faster than
with $\algoC$, but $\algoDplus$ seems to use much more space for this task.
The space comparison, however, is not fair because $\algoD$ also constructs
the wavelet tree of the document array and two select data structures
for the wavelet tree of the $\BWT$ to calculate $\Psi$ values.
These data structures are important for searches on the graph,
but they are superfluous in the construction of the explicit graph.
So in fact $\algoDplus$ uses only a little more space for this task because
the implicit representation of the graph, which must be kept in main memory,
is rather small. Table \ref{tab-results-space-details-50} contains
a detailed breakdown of the space usage of the variants of algorithm $\algoD$.
As the explicit compressed de Bruijn graph, the combination of the implicit
graph and the FM-index supports a graph traversal (albeit in backward
direction). For this task the implicit graph and the FM-index use
much less space than the explicit graph. In contrast to the explicit graph,
our new data structure allows to search for a pattern $P$ in the graph
and to answer questions like: In how many sequences does $P$ occur?
It is this new functionality (notably the document array) that increases
the memory usage again; cf.\ Table \ref{tab-results-space-details-50}.
Despite this new functionality, the overall space consumption of $\algoD$ is
in most cases less than that of $\algoC$; see Table \ref{tab-results-runtime}.

In our next experiment, we measured how long it takes to find the nodes
in the graph that correspond to a pattern $P$. Since the median protein
length in \emph{E.\ coli} is $278$ and a single amino acid is coded by three
nucleotides, we decided to use a pattern length of $900$.
Table \ref{tab-results-search-match} shows the results for $10,000$ patterns
that occur in the pan-genome (if patterns do not occur in the pan-genome,
the search will be even faster; data not shown). Furthermore, we measured
how long it takes to determine to which sequences each node belongs
(using the procedure $getIntervals$ on the wavelet tree of the
document array as described at the end of
Section \ref{sec:Operations on the compressed de Bruijn graph}).
Table \ref{tab-results-docs-all} shows the results for the nodes
corresponding to $10,000$ patterns that occur in the pan-genome.

Finally, we determined the length of the longest string corresponding
to a node in the compressed de Bruijn graph. This is important
because the worst-case search time depends on this length; see end of Section
\ref{sec-Construction of the explicit compressed de Bruijn graph}.
The results can be found in Table \ref{tab-results-lengths}.

\begin{table}[H]
\begin{center}
\setlength{\tabcolsep}{2mm}
\begin{footnotesize}
\begin{tabular}{llrrr}
\hline
                algo &         part &         62 Ecoli &           7xChr1 &             7xHG \\
\hline \\[-0.5em]
              \algoD &       wt-bwt & 0.42   (23.83\%) & 0.44   (36.23\%) & 0.43   (22.68\%) \\
              \algoD &        nodes & 0.10 \hd(5.94\%) & 0.03 \hd(2.61\%) & 0.04 \hd(2.02\%) \\
              \algoD &       $\BVr$ & 0.16 \hd(8.93\%) & 0.16   (12.86\%) & 0.16 \hd(8.25\%) \\
              \algoD &       $\BVl$ & 0.14 \hd(8.04\%) & 0.14   (11.57\%) & 0.14 \hd(7.42\%) \\
              \algoD &       wt-doc & 0.93   (53.26\%) & 0.45   (36.73\%) & 1.13   (59.63\%) \\
[0.5em] \hline \\[-0.5em]
           \algoDone &       wt-bwt & 0.42   (28.57\%) & 0.44   (47.83\%) & 0.43   (26.85\%) \\
           \algoDone &        nodes & 0.10 \hd(7.12\%) & 0.03 \hd(3.44\%) & 0.04 \hd(2.39\%) \\
           \algoDone &       $\BVr$ & 0.00 \hd(0.23\%) & 0.00 \hd(0.12\%) & 0.00 \hd(0.09\%) \\
           \algoDone &       $\BVl$ & 0.00 \hd(0.23\%) & 0.00 \hd(0.12\%) & 0.00 \hd(0.08\%) \\
           \algoDone &       wt-doc & 0.93   (63.85\%) & 0.45   (48.49\%) & 1.13   (70.59\%) \\
[0.5em] \hline \\[-0.5em]
           \algoDtwo &       wt-bwt & 0.16   (13.03\%) & 0.22   (31.01\%) & 0.22   (15.62\%) \\
           \algoDtwo &        nodes & 0.10 \hd(8.67\%) & 0.03 \hd(4.55\%) & 0.04 \hd(2.76\%) \\
           \algoDtwo &       $\BVr$ & 0.00 \hd(0.28\%) & 0.00 \hd(0.16\%) & 0.00 \hd(0.10\%) \\
           \algoDtwo &       $\BVl$ & 0.00 \hd(0.28\%) & 0.00 \hd(0.16\%) & 0.00 \hd(0.10\%) \\
           \algoDtwo &       wt-doc & 0.93   (77.74\%) & 0.45   (64.11\%) & 1.13   (81.42\%) \\
\end{tabular}
\end{footnotesize}
\end{center}
\caption{The first column shows the algorithm used in the experiment
(the $k$-mer size is $50$). The second column specifies the different
data structures used: wt-bwt stands for the wavelet tree
of the $\BWT$ (including rank and select support), nodes stands for
the array of nodes (the implicit graph representation), $\BVr$ and
$\BVl$ are the bit vectors described in
Section \ref{sec-Computation of right-maximal k-mers}
(including rank support), and wt-doc stands for the wavelet tree
of the document array.
The remaining columns show the memory usage in bytes per
base pair and, in parentheses, their percentage.}
\label{tab-results-space-details-50}
\end{table}

\begin{table}[H]
\begin{center}
\setlength{\tabcolsep}{2mm}
\begin{footnotesize}
\begin{tabular}{rlrrr}
\hline
       $k$ &                      &        62 Ecoli &          7xChr1 &            7xHG \\
\hline \\[-0.5em]
        50 &               \algoD &      3     (1.81) &      9     (1.28) &      9     (1.96) \\
        50 &            \algoDone &      3     (1.52) &      9     (0.98) &     11     (1.66) \\
        50 &            \algoDtwo &      6     (1.20) &     20     (0.70) &     29     (1.39) \\
[0.5em] \hline \\[-0.5em]
       100 &               \algoD &      3     (1.78) &     12     (1.26) &     27     (1.94) \\
       100 &            \algoDone &      3     (1.49) &     15     (0.97) &     19     (1.64) \\
       100 &            \algoDtwo &      6     (1.17) &     31     (0.68) &     51     (1.37) \\
[0.5em] \hline \\[-0.5em]
       500 &               \algoD &      9     (1.73) &     20     (1.26) &     22     (1.93) \\
       500 &            \algoDone &     12     (1.43) &     24     (0.96) &     27     (1.63) \\
       500 &            \algoDtwo &     17     (1.11) &     55     (0.67) &     74     (1.36) \\
\end{tabular}
\caption{The first column shows the $k$-mer size
and the second column specifies the algorithm used in the experiment.
The remaining columns show the run-times in seconds for finding the nodes
corresponding to $10,000$ patterns of length $900$ (that occur in the
pan-genome) and, in parentheses, the maximum main memory usage in bytes per
base pair for the data sets described in the text.}
\label{tab-results-search-match}
\end{footnotesize}
\end{center}
\end{table}

\begin{table}[H]
\begin{center}
\setlength{\tabcolsep}{2mm}
\begin{footnotesize}
\begin{tabular}{rlrrr}
\hline
       $k$ &                      &        62 Ecoli &          7xChr1 &            7xHG \\
\hline \\[-0.5em]
        50 &               \algoD &  10.84     (1.81) &   3.31     (1.28) &  15.33     (1.96) \\
        50 &            \algoDone &  10.91     (1.52) &   3.17     (0.98) &  14.88     (1.66) \\
        50 &            \algoDtwo &  11.02     (1.20) &   3.07     (0.70) &  13.02     (1.39) \\
[0.5em] \hline \\[-0.5em]
       100 &               \algoD &   8.31     (1.78) &   2.72     (1.26) &  10.99     (1.94) \\
       100 &            \algoDone &   8.11     (1.49) &   2.83     (0.97) &   9.10     (1.64) \\
       100 &            \algoDtwo &   8.23     (1.17) &   2.84     (0.68) &   9.25     (1.37) \\
[0.5em] \hline \\[-0.5em]
       500 &               \algoD &   2.43     (1.73) &   1.32     (1.26) &   4.51     (1.93) \\
       500 &            \algoDone &   2.78     (1.43) &   1.32     (0.96) &   4.22     (1.63) \\
       500 &            \algoDtwo &   2.32     (1.11) &   1.29     (0.67) &   4.30     (1.36) \\
\end{tabular}
\caption{The first column shows the $k$-mer size
and the second column specifies the algorithm used in the experiment.
The remaining columns show the run-times in seconds for
finding out to which sequences each of the nodes belongs (where
the nodes correspond to $10,000$ patterns of length $900$ that occur in the
pan-genome) and, in parentheses, the maximum main memory usage in bytes per
base pair for the data sets described in the text.}
\label{tab-results-docs-all}
\end{footnotesize}
\end{center}
\end{table}

\begin{table}[H]
\begin{center}
\setlength{\tabcolsep}{2mm}
\begin{footnotesize}
\begin{tabular}{rrrr}
\hline
       $k$ &     62 Ecoli &       7xChr1 &         7xHG \\
\hline \\[-0.5em]
        50 &       79,967 &       41,571 &       36,579 \\
       100 &      173,366 &       85,773 &      203,398 \\
       500 &      179,671 &    2,283,980 &    1,402,896 \\
\end{tabular}
\caption{The first column specifies the $k$-mer size
and the remaining columns show the length of the longest
string corresponding to a node in the compressed de Bruijn graph.}
\label{tab-results-lengths}
\end{footnotesize}
\end{center}
\end{table}

\section{Conclusions}
\label{sec-conclusions}
We have presented a space-efficient method to build the compressed de Bruijn
graph from scratch. An experimental comparison with splitMEM showed
that our algorithm is more than an order of magnitude faster than splitMEM
while using significantly less space (two orders of magnitude). To
demonstrate its scalability, we successfully applied it to seven complete
human genomes. Consequently, it is now
possible to use the compressed de Bruijn graph for much larger pan-genomes
than before (consisting e.g.\ of hundreds or even thousands of different
strains of bacteria). Moreover, the combination of the implicit graph
and the FM-index can be used to search for a pattern $P$ in the graph
(and to traverse the graph).

Future work includes a parallel implementation of the construction algorithm.
Moreover, it should be worthwhile to investigate the time-space trade-off
if one uses data structures that are optimized for highly repetitive texts;
see \cite{NAV:ORD:2014} and the references therein.

\section*{Acknowledgments}
This work was supported by the DFG (OH 53/6-1).

\end{document}